\documentclass[12pt]{article}
\usepackage{amsmath,amssymb,amsbsy,amsthm}

\let\epsilon\varepsilon
\let\phi\varphi

\def\N{\mathbb N}
\def\R{\mathbb R}

\newtheorem{theorem}{Theorem}

\newtheorem{remark}{Remark}
\newtheorem{proposition}{Proposition}

\newtheorem{definition}{Definition}

\begin{document}
\title{Using Data Compressors to Construct Rank Tests}
\date{}
\author{Daniil Ryabko, J\"urgen Schmidhuber\\
\normalsize IDSIA, Galleria 2, CH-6928\ Manno-Lugano, Switzerland \\
\normalsize \{daniil, juergen\}@idsia.ch}

\maketitle
\begin{abstract}
Nonparametric rank tests for homogeneity and component independence are proposed, which are based on data compressors. 
For homogeneity testing the idea is to compress the binary string obtained by ordering the two joint samples 
and writing 0 if the element is from the first sample and 1 if it is from the second sample and breaking ties by randomization (extension 
to the case of multiple samples is straightforward).
$H_0$ should be rejected if the string is compressed (to a certain degree) and accepted otherwise. We show that such a test
obtained from an ideal data compressor is valid against all alternatives.

Component independence is reduced to homogeneity testing by constructing two samples, one of which is the first half of the original
and the other is the second half with one of the components randomly permuted.
\end{abstract}

\section{Introduction}
We consider two classical problems of mathematical statistics.
The first one is homogeneity testing:
two (or more; see below) samples $X_1,\dots X_n$ and $Y_1,\dots,Y_n$ with elements in $\R$ are given. 
It is assumed that the elements are drawn independently and within samples the distribution is the same.
We want to test the hypothesis $H_0$ that $X_i$ and $Y_i$ are distributed according to the same
distribution versus $H_1$ that the distributions generating the samples are different. This is called homogeneity testing. 
Absolutely no assumptions are made on the distributions.

The second one is component independence: a sample $Z_1,\dots, Z_n$ is given, generated i.i.d. according to some
distribution $F_Z$.   Each element $Z_i$ consists of two (or more) components $Z^1_i$ and $Z^2_i$. We wish 
to test whether the components are independent of each other. That is, $H_0$ is that the marginal distributions
are independent whereas $H_1$ is that there is some dependency. Again, no assumption is made on the distribution $F_Z$. 
 
Both problems are well-known problems of nonparametric mathematical statistics.
For example,  a classical test for homogeneity is Kolmogorov-Smirnov test (which assumes, however, that
the distributions generating the samples are continuous). There are many other nonparametric tests; some of 
the tests use {\em ranks} of elements within the joint sample, instead of using the actual samples. Such is,
for example, Wilcoxon's test, see \cite{leh} for an overview  (which also makes some additional assumptions on the distribution). 

In this work we present simple nonparametric (distribution-free) rank tests for homogeneity and component independence based
on data compressors.

The idea to use real-life data compressors for testing classical statistical hypotheses,
such as homogeneity,  component independence and some others, was suggested in \cite{r1,r2}.
In these works statistical tests based on data compressors  are constructed which fall into 
the classical framework of nonparametric mathematical statistics, in particular, the Type I error is fixed 
while Type II error goes to 0 under a wide range of alternatives. The hypotheses considered there 
mostly concern data samples drawn from {\em discrete} (e.g.\  finite) spaces. Some tests for
continuous spaces are also proposed based on partitioning. Here we extend this approach to {\em rank}
tests, allowing testing homogeneity and component independence without the need of partitioning the sample spaces and making them finite.
The idea of using data compressors for tasks other than actual data compression 
was suggested in \cite{v1,v2,v3}, where data compressors are applied to such tasks as classification and clustering.
These works were largely inspired by Kolmogorov complexity, which is also an important tool for the present work.

An ``ideal'' data compressor is the one that compresses its input up to its Kolmogorov complexity. 
This is intuitively obvious since, informally, Kolmogorov complexity of a string is the length of the shortest 
program that outputs this string.  Such data compressors do not exist; in particular, Kolmogorov complexity itself
is incomputable. Real data compressors, however, can be considered as approximations of ideal ones.

In this work we provide a simple empirical procedure for testing homogeneity and component independence 
with data compressors; we show that for an ideal data compressor this  procedure 
provides a statistical test  which is valid against all alternatives (Type II error goes to zero);
while Type I error is guaranteed to be below a pre-defined level (so-called significance level) for all 
data compressors, not only for ideal ones. 
It should also be noted that the theoretical assumption underlying  data compressors used in real life is that the data to  compress is {\em stationary}.
Thus the tests designed in \cite{r1,r2} are provably valid against any stationary and ergodic alternative, while 
these tests are based on real data compressors,  not only on ideal ones.
In our case, the alternative arising in rank test under $H_1$ is not stationary. Thus we prove theorems only about ideal data compressors, 
and real data compressors can be used heuristically. However, it can be conjectured that the same results
can be proven for some particular real-life data compressors, for example for those which are based on the measure 
$R$ from \cite{r88} or on the LZ algorithm \cite{lz}.

\section{Homogeneity testing}
Homogeneity testing is the following task. Let there be given two samples $X=\{X_1,\dots X_m\}$ and $Y=\{Y_1,\dots,Y_k\}$ 
(the case of more than two samples will also be considered).
$X_i$ are drawn independently according to some probability distribution $F_X$ on $\R^d$ ($d\in\N$) and 
$Y_i$ are drawn independently from each other and from $X_i$ according to some distribution $F_Y$ on $\R^d$.
The goal is to test whether $F_X = F_Y$. No assumption is made on the distributions 
$F_X$ and $F_Y$; we only assume that $X_i$ and $Y_i$ are drawn independently within the samples and jointly.
So, we wish to test the hypothesis $H_0= \{ (F_X,F_Y): F_X= F_Y \}$  against $H_1= \{ (F_X,F_Y): F_X\ne F_Y \}$.

A {\em code} $\phi$ is a function $\phi: B^*\rightarrow B^*$ from the set of all finite words over binary  alphabet $B=\{0,1\}$ to itself,
such that $\phi$ is an injection (that is, $a\ne b$ implies $\phi(a)\ne \phi(b)$ for $a,b\in B^*$).
A trivial example of a code is the identity $\phi_{id} (a) = a$. Less trivial examples that we have 
in mind are data compressors, such as  \texttt{zip, rar, arj,} or others, which take a word and output a ``compressed'' version of it 
(which in fact is often longer than the original)  
from which the original input can always be recovered.
We will construct (reasonable) tests for homogeneity  from (good) data compressors.

First let us  assume that $\bf d=1$ (that is, $X_i, Y_i\in\R$). 
Let $Z_1\le Z_2 \le \dots\le Z_{m+k}$ denote the joint sample 
constructed by ordering jointly two samples $X$ and $Y$. Construct the word $A=A_1\dots,A_{m+k}$ as
follows:  for each $i$ $A_i=0$ if $Z_i$ is taken from the sample $X$ ($Z_i\in X$) and $A_i=1$ if $Z_i$ is from the sample $Y$ ($Z_i\in Y$) where
ties are broken by randomization: if $Z_j=Z_{j+1}=\dots Z_{j'}$ 
and there are $m'$  elements of the sample $X$ which are equal to $Z_j$ and $k'$ elements of the sample 
$Y$ which are equal to $Z_j$ then the word $A_j\dots A_{j'}$ is chosen randomly from all $\frac {(m'+k')!} {m'! k'!}$
 binary words which have $m'$ zeros and $k'$ ones, assigning equal probabilities to all words.

Now consider the case $\bf d>1$, that is, the elements of the samples $X$ and $Y$ are from $\R^d$, $d>1$.
Construct samples $\bar X=\bar X_1,\dots,\bar X_m$ and  $\bar Y=\bar Y_1,\dots,\bar Y_m$ as follows:
$\bar X_t:=x^{11}_t,x^{21}_t,\dots,x^{d1}_t,x^{12}_t,x^{22}_t,\dots,x^{d2}_t,\dots$ where
$x^{ij}_t$ is the $j$th element in the binary expansion of the $i$th component of $X_t$ (in case the expansion is ambiguous always take the one with more zeros), and analogously for $Y$.
 Denote the
described function which converts $X$ to $\bar X$  by $\tau$. 
Construct the string $A$ applying the (single--dimensional) procedure described above to the samples $\bar X$ and $\bar Y$.

  Let $|K|$ denote the length of a string $K$.

\begin{definition}[Homogeneity  test $G_\phi$] For any code $\phi$ the test for homogeneity  $G_\phi$ is constructed as follows. 
It rejects the hypothesis $H_0$ (outputs {\em reject}) 
at the level of significance $\alpha$ if 
\begin{equation}\label{eq:test}
|\phi(A)| \le \log \alpha N
\end{equation}
where $N:=\frac{(m+k)!}{m! k!}$ and $\log$ is base 2,
and accepts  $H_0$ 
(outputs {\em accept}) 
otherwise.
\end{definition}
\begin{definition}[More than two samples] In case we are given $r$ samples where $r\ge 2$ and wish to test
$H_0$ that they all are generated according to the same distribution versus at least two distributions
are different, the test is the same, except for that the string $A$ is not binary but from $r$-element alphabet
and in the test above instead of  $N$  take
$$
N':= \frac{ ( \sum_{i=1}^r m_i)!}{\prod_{i=1}^r m_i},
$$
where $m_i$ are the sizes of the samples.
\end{definition}

The intuition is as follows. 
Observe that if the distributions $F_X$ and $F_Y$ are equal (that is, $H_0$ is true), 
then the string $A$ is just a random 
binary string with $m$ zeros and $k$ ones; all such strings have equal probabilities under $H_0$. 
 Thus a good data compressor should  be able to compress 
it to about $\log N$ bits, but no code can compress many such strings to less than $\log N-t$ bits ($t>0$), 
since there are $N$ such strings and only $2^{-t}N$ binary strings of length  $\log N-t$.
 
\begin{proposition}[Type I error]\label{prop:t1} Let $d=1$.
 For any code $\phi$ and any $\alpha\in[0,1]$ the Type I error of the test $G_\phi$ with level of significance $\alpha$ is not greater than $\alpha$:
\begin{equation}\label{eq:prop1}
P\{ X,Y: G_\phi (X,Y) = reject \} \le \alpha 
\end{equation}
 for all $ P=(F_X, F_Y)\in H_0$.
\end{proposition}
\begin{remark}
The proposition still holds if $H_0$ is rejected when 
\begin{equation}\label{eq:t} 
 |\phi(A)| \le  (k+m)h\left({k\over k+m}\right) + \log\alpha -\log (k+m),
\end{equation}
 where
$h(t)$ is the entropy 
\begin{equation}\label{eq:ent} 
  h(t):= - t \log t -  (1-t)  \log (1-t).
\end{equation}
In case of $r$ samples~(\ref{eq:t}) takes the form 
\begin{equation}\label{eq:t2} 
 |\phi(A)| \le  n h \log r  + \log\alpha -\log n
\end{equation}
 with $n=\sum_{i=1}^r m_i$ and
$ 
 h= -\sum_{i=1}^r \frac{m_i}{n} \log\frac{m_i}{n}.
$
\end{remark}
\begin{proof}
As it was noted, under $H_0$ for every string $a\in B^{k+m}$ such that $a$ consists of $m$ zeros and $k$ ones  $P(A=a)=1/N$
(that is, all such strings are equiprobable). Since there are only $\alpha N$ binary strings of length $ \log \alpha N$ and $\phi$ is an injective function,
that is each codeword is assigned to at most one word,
we get $ P\{ X,Y: |\phi(A)|  \le  \log \alpha N \} \le {1\over N} N\alpha =\alpha$ which together with the definition of $G_\phi$ implies~(\ref{eq:prop1}).

The statement of the Remark can be derived from Stirling's expansion for $N$ and $N'$.
\end{proof}
\begin{remark}
The term $-\log(k+m)$ in~(\ref{eq:t}) is due to the fact that there are only $(m+k)!\over m!k!$ strings with $m$ zeros and $k$ ones (among
$2^{k+m}$ all binary strings of this length).
 So the code $\phi$ can specifically assign shorter codewords to these strings.
As real data compressors are not designed to favour  strings of this particular ratio of zeros and ones, in practice it is recommended to omit the term 
$-\log(k+m)$ in~(\ref{eq:t}). The same concerns the term $-\log n$ in~(\ref{eq:t2}).
\end{remark}

Obviously, for some codes the test is useless (for example if $\phi$ is the identity mapping) and Proposition~\ref{prop:t1} is
only useful when the Type II error goes to zero.
Next we will define ``ideal'' codes (the codes that compress a word up to its Kolmogorov complexity)
 and show that for them indeed the probability of {\em accept} goes to zero under any distribution in $H_1$. 

Informally, Kolmogorov complexity of a string $A$ is the length of the shortest program that outputs  $A$ (on the empty input).
Clearly, the best, ``ideal'', data compressor can compress any string $A$ up to its Kolmogorov complexity, and not more (except may be for a constant).
Next we  present a definition of Kolmogorov complexity; for fine
details see \cite{ks, livi}.
The complexity of  a string $A\in B^*$ with respect to a Turing machine $\zeta$
is defined as
$$
 C_\zeta(A)=\min_p\{l(p):\zeta(p)=A\},
$$
where $p$ ranges over all binary strings (interpreted as  programs for $\zeta$; minimum over empty set is defined as $\infty$).
There exists a Turing machine $\zeta$ such that $C_\zeta(A)\le C_{\zeta'}(A)+c_{\zeta'}$
for any $A$ and any  Turing machine $\zeta'$ (the constant $c_{\zeta'}$
depends on $\zeta'$ but not on $A$).
Fix any such $\zeta$ and define
 \emph{Kolmogorov complexity} of a string $A\in\{0,1\}^\infty$ as
$
C(A):=C_\zeta(A).
$
Clearly, $C(A)\le |A|+b$ for any $A$ and for some $b$ depending only on $\zeta$.

\begin{definition}[ideal codes]
Call a code $\phi$ {\em ideal} if some constant $c$ the equality $|\phi(A)|\le C(A)+c$ holds for any binary string $A$.  
\end{definition}
Clearly such codes exist. 

\begin{proposition}[Type II error: universal validity]\label{prop:p2} For any ideal code $\phi$ Type~II error of the test $G_\phi$ with any fixed  significance level $\alpha>0$ 
goes to zero $P\{ X,Y: G_\phi (X,Y) = accept \}\rightarrow0$ for any $P$ in $H_1$ if $k,m\rightarrow\infty$ in such a way that $0<a<{k\over m}<b<1$ for some $a,b$.
\end{proposition}
\begin{proof}
First observe that the function $\tau$ that converts  $d$-dimensional  samples $X$ and $Y$ to single-dimensional samples $\bar X$ and $\bar Y$ 
has the following properties: if $X$ and $Y$ are distributed according to different distributions then $X$ and $Y$ are also distributed
according to different distributions. Indeed, $\tau$ is  one to one, and transforms cylinder sets (sets of the form
$\{x\in\R^d: x^{i_1j_1}=b_1,\dots,x^{i_tj_t}=b_t; b_l\in \{0,1\}, t,i_l,j_l\in\N (1\le l\le t)\}$) to cylinder sets. So together with $F_X$ ($F_Y$) 
it defines some distribution $F_{\bar X}$  ($F_{\bar Y}$) on $\R$.
If distributions $F_X$  and $F_Y$ are different then they are different on some cylinder set $T$, but then $F_{\bar X}(\tau (T))\ne F_{\bar Y}(\tau (T))$.
Thus further in the proof we will assume that $d=1$.

We have to show that Kolmogorov complexity $C(A)=|\phi(A)|$ of the string $A$ is less than $\log \alpha N \ge  (k+m)h\left({k\over k+m}\right) + \log\alpha -\log (k+m)$ for any fixed $\alpha$
from some $k,m$ on. To show this, we have to find a sufficiently short description $s(A)$ of the string $A$; then the Kolmogorov complexity $|\phi(A)|$ is not greater
than $|s(A)|+c$ where $c$ is a constant. 

If $H_1$ is true then $F_X\ne F_Y$ and so there exist some interval $T=(-\infty,t]$ and some $\delta>0$ such that $|F_X(T) -  F_Y(T)|>2\delta$. 
Then we will have 
\begin{equation}\label{eq:f1}
 \left|\frac {\#\{x\in X\cap T\}}{m} - \frac {\#\{y\in Y \cap T\}}{k}\right| >\delta
\end{equation}
 from some $k,m$ on with probability 1. 

Let $A'$ be the starting part of $A$ that consists of all elements that belong to $T$ and let $m':={\#\{x\in X\cap T\}}$ and $k':={\#\{y\in Y\cap T \}}$.
A description of $A'$ can be constructed as the index of $A'$ in the set (ordered, say, lexicographically) of all binary strings
of length $m'+k'$ that have exactly $m'$ zeros and $k'$ ones plus the description of $m'$ and $k'$. Thus the length of such a description is bounded 
by $\log \frac{(k'+m')!}{k'!m'!} \le (m'+k') h(\frac{k'}{m'+k'})$ plus $\log k' +\log m' +const$ (the inequality follows from $n!\le n^n$ for all $n$).
Let $\bar A$ denote the remaining part of $A$ (that is, what goes after $A'$). The length of the description of $\bar A$ 
is bounded by $(\bar m+\bar k) h(\frac{\bar k}{\bar m+\bar k}) + \log \bar k +\log \bar m +const$ where $\bar m = m-m'$ and $\bar k = k - k'$.
Since $h$ is concave and $\frac{k}{m+k}$ is between $\frac{k'}{m'+k'}$ and $\frac{\bar k}{\bar m+\bar k}$, from Jensen's inequality we obtain 
$$ 
h\left(\frac{k}{m+k}\right) - \left(\frac{m'+k'}{m+k} h\left(\frac{k'}{m'+k'}\right) + \frac{\bar m+\bar k}{m+k} h\left(\frac{\bar k}{\bar m+\bar k}\right)\right) >0.
$$
Denote this difference by $\gamma(k,m,k',m')$. Let $\gamma =\inf \gamma(k,m,k',m')$ where the infimum is taken over all pairs $k,m$ that satisfy the condition 
of the proposition $0<a<{k\over m}<b<1$ and $k',m'$ that satisfy~(\ref{eq:f1}). It follows that $\inf|{k'\over m'}-{k\over m}|>0$ and 
$\inf|{\bar k\over \bar m}-{k\over m}|>0$. Thus, $\gamma$ is positive and depends only on $a,b$ and $\delta$.
To uniquely describe $A$ we need the description of $A'$ and  $\bar A$ and also $k$ and $m$; these have to be encoded in a self-delimiting way;
the length of such a description $s(A)$ is bounded by the lengths of description of $A'$, $\bar A$ plus $\log (k+m)$ and some constant.
Thus 
\begin{multline*}
(k+m)h\left({k\over k+m}\right) + \log\alpha -\log (k+m) - |\phi(A)|
  \ge  
 \\ (k+m)h\left({k\over k+m}\right) + \log\alpha -2 \log (k+m) \\ 
 - 
  \left(\frac{m'+k'}{m+k} h\left(\frac{k'}{m'+k'}\right) + \frac{\bar m+\bar k}{m+k} h\left(\frac{\bar k}{\bar m+\bar k}\right)\right) -c
 \\ \ge (k+m)\gamma - 2\log(k+m) -c
\end{multline*}
for some constant $c$; clearly, this expression is greater than $0$ from some $k,m$ on.
\end{proof}

So, as a corollary of Propositions~\ref{prop:t1} and~\ref{prop:p2} we get the following statement.
\begin{theorem} 
For any code $\phi$ and any $\alpha\in(0,1]$ the Type~I error of the test $G_\phi$ with level of significance $\alpha$ is not greater than $\alpha$.
If, in addition, the code $\phi$ is ideal then the Type~II of $G_\phi$ error tends to 0 as the sample size $n$ approaches infinity.
\end{theorem}

\section{Component independence testing}
Component independence testing is the following task.
A sample $Z=Z_1,\dots,Z_n$ is given where each $Z_i$ consists of $r$  components $Z^1_i,Z^2_i,\dots,Z^r_i$, $Z^j_i\in \R^{d_j}$.
The sample is generated according to some probability distribution $F_Z$ on $\R^d$, where $d:=\sum_{j=1}^rd_j$.
The goal is to test whether the components are distributed independently. That is, 
$H_0$ is that 
\begin{equation}\label{eq:h0ce}
F_Z(Z^1_1\in T_1,\dots, Z^r_1\in T_r) = \prod_{j=1}^r F_Z(Z^j_1\in T_j)
\end{equation} for all measurable $T_j\subset \R^{d_j}$, $1\le j\le r$. 
$H_1$ is the negation of $H_0$ (the equality~(\ref{eq:h0ce}) is false for some selection of the sets $T_j$, $1\le j\le r$). Again, no assumption is made on the form of the distribution $F_Z$.

Fix any code $\phi$ and construct the test for component independence $I_\phi$ as follows.
Assume that $n=2m$ for some $m$ and define the samples $X$ and $\bar Y$ as the first and the second half of the sample $Z$: $X_1=Z_1,\dots, X_{m}=Z_m$ and
$\bar Y_1=Z_{m+1},\dots,\bar Y_m=Z_{2m}$ (if $n$ is odd then make samples $X$ and $\bar Y$ of sizes $[n/2]$ and $n-[n/2]$).
Construct the sample $Y$ from $\bar Y$ by permuting the components  independently: $Y^j_i=\bar Y^j_{\pi_j(i)}$,   $1\le i\le m$, $1\le j\le r$ 
where $\pi_j$ are permutations $1\dots m$, selected at random (with equal probabilities) independently of each other.

\begin{definition}[Component independence test $I_\phi$]
The test $I_\phi$ (with level of significance $\alpha$)  consists in application of the test for homogeneity $G_\phi$ to the samples $X$ and $Y$ (with level of significance $\alpha$).
\end{definition}

Indeed, it is easy to check that $H_0$ is true if and only if $X$ and $Y$ are distributed according to the same distribution.
So we get the following statement.
\begin{theorem} 
For any code $\phi$ and any $\alpha\in(0,1]$ the Type~I error of the test $I_\phi$ with level of significance $\alpha$ is not greater than $\alpha$.
If, in addition, the code $\phi$ is ideal then the Type~II error of $I_\phi$ error tends to 0 as the sample size $n$ approaches infinity.
\end{theorem}

\end{document}